\documentclass[conference]{IEEEtran}
\makeatletter
\def\ps@headings{%
\def\@oddhead{\mbox{}\scriptsize\rightmark \hfil \thepage}%
\def\@evenhead{\scriptsize\thepage \hfil \leftmark\mbox{}}%
\def\@oddfoot{}%
\def\@evenfoot{}}
\makeatother 
\usepackage{color}
\usepackage{citesort}
\usepackage{graphicx}
\usepackage{subfigure}
\usepackage{amsmath}
\usepackage{amsfonts,amssymb}
\usepackage{epsf}

\def\inline#1:{\par\vskip 7pt\noindent{\bf #1:}\hskip 10pt}

\long\def\comment #1\commentend{}
\long\def\commabs #1\commabsend{}
\long\def\commful #1\commfulend{#1}

\def\greedya{\mathtt{Greedy}}
\def\greedyb{\mathtt{Greedy\_shortest}}
\def\competa{\mathtt{BatchAll}}
\def\competb{\mathtt{BatchLim}}
\def\competc{\mathtt{BatchAllDisp}}
\def\competd{\mathtt{BatchLimDisp}}
\def\epsi{\varepsilon}
\newtheorem{theorem}{Theorem}[section]
\newtheorem{lemma}[theorem]{Lemma}

\newtheorem{corollary}[theorem]{Corollary}

\def\blackslug{\hbox{\hskip 1pt \vrule width 4pt height 8pt depth 1.5pt
                 \hskip 1pt}}
\def\QED{\quad\blackslug\lower 8.5pt\null\par}
\renewcommand{\paragraph}[1]{\par\noindent\textbf{#1.}}

\newcommand{\ignore}[1]{}

\title{Throughput Optimal On-Line Algorithms for Advanced
Resource Reservation in Ultra High-Speed Networks}
\author{\authorblockN{Reuven Cohen$^{1,2}$, Niloofar Fazlollahi$^1$ and David Starobinski$^1$}
\authorblockA{$^1$Dept. of Electrical and Computer Engineering\\
Boston University, Boston, MA 02215\\
$^2$Dept. of Physics, MIT, Cambridge, MA 02139\\
Email: reuven@mit.edu,\{nfazl,staro\}@bu.edu}}
\begin{document}

\date{\today}
\maketitle

\begin{abstract}
Advanced channel reservation is emerging as an important feature
of ultra high-speed networks requiring the transfer of large
files. Applications include scientific data transfers and database
backup. In this paper, we present two new, on-line algorithms for
advanced reservation, called $\competa$ and $\competb$, that are
guaranteed to achieve optimal throughput performance, based on
multi-commodity flow arguments. Both algorithms are shown to have
polynomial-time complexity and provable bounds on the maximum
delay for $1+\epsi$ bandwidth augmented networks. The $\competb$
algorithm returns the completion time of a connection immediately
as a request is placed, but at the expense of a slightly looser
competitive ratio than that of $\competa$. We also present a
simple approach that limits the number of parallel paths used by
the algorithms while provably bounding the maximum reduction
factor in the transmission throughput. We show that, although the
number of different paths can be exponentially large, the actual
number of paths needed to approximate the flow is quite small and
proportional to the number of edges in the network. Simulations
for a number of topologies show that, in practice, 3 to 5 parallel
paths are sufficient to achieve close to optimal performance. The
performance of the competitive algorithms are also compared to a
greedy benchmark, both through analysis and simulation.
\end{abstract}

%%%%%%%%%%%%%%%%%%%%%%
\section{Introduction}
%%%%%%%%%%%%%%%%%%%%%%
TCP/IP has been so far
considered and used as the core infrastructure for all sorts of data
transfer including bulk FTP applications. Yet, it has recently been
observed that in ultra high speed networks, there exists a large gap
between the capacity of network links and the maximum end-to-end
throughput achieved by TCP \cite{doereport,RWCW05}. This gap, which
is mainly attributed to the shared nature of Internet traffic, is
becoming increasingly problematic for modern Grid and data backup
applications requiring transfer of extremely large datasets on the
orders of terabytes and more.

These limitations have raised efforts for development of an
alternative protocol stack to complement TCP/IP. This protocol stack
is based on the concept of \emph{advanced reservation}
\cite{RWCW05,CFS06}, and is specifically tailored for large file
transfers and other high throughput applications. The most important
property of advanced reservation is to offer hosts and users the
ability to reserve in advance \emph{dedicated} paths to connect
their resources.
The importance of advanced reservation in supporting dedicated
connections has been made evident by the growing number of
testbeds related to this technology, such as UltraScience
Net~\cite{RWCW05}, On-demand Secure Circuits and Advanced
Reservation Systems (OSCARS)~\cite{oscar}, and Dynamic Resource
Allocation via GMPLS Optical Networks (DRAGON)~\cite{dragon}.

Several protocols and algorithms have been proposed in the
literature to support advance reservation, see,
e.g.,~\cite{RWCW05,GO00,CFS06} and references therein. To the
authors' knowledge, none of them guarantees the same throughput
performance as  an optimal off-line algorithm. Instead, most are
based on greedy approaches, whereas each request is allocated a
path guaranteeing the earliest completion time at the time the
request is placed.

Our first contribution in this paper is to uncover fundamental
limitations of greedy algorithms. Specifically, we show that there
exist network topologies and request patterns for which the maximum
 throughput of these algorithms can be arbitrarily
smaller than the optimal throughput, as the network size grows.

Next, we present competitive advanced reservation algorithms that
provably achieve the maximum throughput for any network and
request profile. The first algorithm, called $\competa$, provides
a competitive ratio on the maximum delay experienced by each
request in an augmented resource network with respect to the
optimal off-line algorithm in the original network. $\competa$
groups requests into batches. In each batch, paths are efficiently
assigned based on a maximum concurrent flow optimization. All the
requests arriving during a current batch must wait until its
completion before being assigned to the next batch.

The $\competa$ algorithm does not return the connection completion
time to the user at the time a request is placed, but only when the
connection actually starts. Our second algorithm, called $\competb$,
returns the completion time immediately as a request is placed, but
at the expense of a slightly looser competitive ratio. This
algorithm operates by limiting the length of each batch.

The presented competitive algorithms are based upon multicommodity
flow algorithms, and therefore can be performed efficiently in
polynomial time in the size of the network for each request.

Our model assumes that the network infrastructure supports path
dispersion, i.e., multiple paths in parallel can be used to route
data. Obviously, a too large path dispersion may be undesirable,
as it may entail fragmenting a file into a large number of
segments and reassembling them at the destination. To address this
issue, we present a simple approach, based on the max-flow min-cut
theorem, that limits the number of parallel paths while bounding
the maximum reduction factor in the transmission throughput. We
prove that this bound is tight. We, then, propose two algorithms
$\competc$ and $\competd$, based upon $\competa$ and $\competb$
respectively. These algorithms perform similarly to the original
algorithms, in terms of the batching process. However, after
filling each batch, the algorithms will limit the dispersion of
each flow. Although these algorithms are not throughput optimal
anymore, they are still throughput competitive.

While the main emphasis of this paper is on the derivation of
theoretical performance guarantees on the maximum throughput of
advanced reservation protocols, we also provide simulation results
illustrating the performance of these protocols in terms of the
average delay. The simulations show that the proposed competitive
algorithms compare favorably to a greedy benchmark, although the
benchmark does sometime exhibit superior performance, especially
in sparse topologies or at low load. With respect to path
dispersion, we show that excellent performance can be achieved
with as few as five or so parallel paths per connection.

Note that preliminary findings leading to this work were reported
in a two-page abstract~\cite{FCS07b}. The only overlap is
Lemma~\ref{lemma:multicomm} and Theorem~\ref{thm:compet}, which
were presented without proof.

This paper is organized as follows. In Section \ref{sec:related}, we
briefly scan the related work on advanced reservation, competitive
approaches,
and path dispersion. 
In Section \ref{sec:greedy}, we introduce our model and notation
used throughout the paper. We also present natural greedy
approaches and demonstrate their inefficiency. In
Section~\ref{sec:comp}, we describe the $\competa$ and $\competb$
algorithms and derive results on their competitive ratios. Our
method for bounding path dispersion is described and analyzed in
Section~\ref{sec:dispersion}. In Section \ref{sec:simulations},
simulation results evaluating the performance of some of our
algorithms under different network topologies and traffic
parameters are presented. We conclude the paper in Section
\ref{sec:summary}.

%%%%%%%%%%%%%%%%%%%%%%%%%%%%%%%%%%%%%%%%%%%%%%%
\section{Related Work}
\label{sec:related}
%%%%%%%%%%%%%%%%%%%%%%%%%%%%%%%%%%%%%%%%%%%%%%%
Competitive algorithms for advanced reservation networks is the focus of
\cite{tardos}. This work discusses the lazy ftp problem, where
reservations are made for channels for the transfer of different files.
The algorithm presented there provides a 4-competitive algorithm for the
makespan (the total completion time). However, \cite{tardos} focuses on the
case of fixed routes. When routing is also to be considered, the time complexity
of the algorithm presented there may be exponential in the network size.

Another recent work~\cite{ST05}, focusing on routing in packet
switched network in an adversarial setting, discusses choosing
routes for fixed size packets injected by an adversary. It
enforces regularity limitations on the adversary which are
stronger than the ones required here, and achieves the network
capacity with a guarantee on the maximum queue size. It does not
discuss the case of advance reservation with different job sizes
or bandwidth requirements. It is based upon approximating an
integer programming, which may not be extensible to a case where
path reservation, rather than packet-based routing is involved.

Most of the works regarding competitive approaches to routing
focused mainly on call admission, without the ability of advance
reservation. For some results in this field see, e.g.,
\cite{ANR92,AAFLS01}. Some results involving advanced reservation
are presented in~\cite{AdmissionCtrl}. However, the path selection
there is based on several alternatives supplied by a user in the
request rather than a path selection using an automated mechanism
attempting to optimize performance, as discussed here. In
\cite{Plotkin95} a combination of call admission and circuit
switching is used to obtain a routing scheme with a logarithmic
competitive ratio on the total revenue received. A competitive
routing scheme in terms of the number of failed routes in the
setting of ad-hoc networks is presented in \cite{AwerbuchHRK05}. A
survey of on-line routing results is presented in
\cite{Leonardi96}.  A competitive algorithm for admission and
routing in a multicasting setting is presented in \cite{GoelHP05}.
Most of the other existing work in this area consists of heuristic
approaches  which main emphasis are on the algorithm correctness
and computational complexity, without throughput guarantees.

In \cite{GKS03} a rate achieving scheme for packet switching
at the switch level is presented. Their scheme is
based on convergence to the optimal
multicommodity flow using delayed decision for queued packets.
Their results somewhat resemble our
$\competa$ algorithm. However, their scheme depends on
the existence of an average packet size, whereas our scheme
addresses the full routing-scheduling question for any size
distribution and any (adversarial) arrival schedule. In
\cite{WCM06} queuing analysis of several optical transport
network architectures is conducted, and it is shown that
under some conditions on the arrival process, some of the
schemes can achieve the maximum network rate. As the previous one, this
paper does not discuss the full routing-scheduling question discussed
here and does not handle unbounded job sizes.
Another difference is that our paper provides an algorithm, $\competb$
discussed below, guaranteeing the ending time of a job at the time of
arrival, which, as far as we know is the first such algorithm.

Many papers have discussed the issue of path dispersion and
attempted to achieve good throughput with limited dispersion, a
survey of some results in this field is given in \cite{GK97}. In
\cite{SSC02,CLL01} heuristic methods of controlling multipath
routing and some quantitative measures are presented. As far as we
know, our work proposes the first formal treatment allowing the
approximation of a flow using a limited number of paths with any
desired accuracy.

% ----------------------------------------------------------------
\section{Greedy Algorithms and their Limitations}
\label{sec:greedy}
% ----------------------------------------------------------------

\subsection{Network Model}

We first present the network model that will be used throughout
the paper (both for the greedy and competitive algorithms). The
model consists of a general  network topology, represented by a
graph $G(V,E)$, where $V$ is the set of nodes and $E$ is the set
of links connecting the nodes. The graph $G$ can be directed or
undirected. The capacity of each link $e \in E$ is $C(e)$.
A connection request, also referred to as job, contains the tuple $(s,d,f)$, where $s \in V$ is the
source node, $d \in V-\{ s \}$ is the destination node, and $f$ is
the file size.

 Accordingly, an advance reservation algorithm
computes a starting time at which the connection can be initiated, a
set of paths used for the connection, and an amount of bandwidth
allocated to each path. Our model supports path dispersion, i.e.,
multiple paths in parallel can be used to route data (in section
\ref{sec:dispersion-alg}, we discuss practical methods to bound path
dispersion while still providing performance guarantees). In
addition, the greedy algorithms introduced below employ path
switching, whereby a connection is allowed to switch between
different paths throughout its duration.

In subsequent sections, we will make frequent use of multicommodity
functions. The multicommodity flow problem is a linear planning
problem returning $\mathtt{true}$ or $\mathtt{false}$ based upon the
feasibility of transmitting concurrent flows from all pairs of
sources and destination during a given time duration, such that the
total flow through each link does not exceed its capacity. It is
solved by a function $\mathtt{multicomm}(G,L,T)$, where $L$ is a
list of jobs, each containing a source, a destination, and a file
size, and $T$ is the time duration.

In the sequel of this paper, we will be interested in
characterizing the saturation throughput of the various
algorithms, which is defined as the maximum arrival rate of
requests such that the average delay experienced by requests is
still finite (delay is defined as the time elapsing between the
arrival of a connection request until the completion of the
corresponding connection).

The maximum concurrent flow is calculated by the function
$\mathtt{maxflow}(G,L)$. It returns $T_\mathrm{min}$, the minimum
value of $T$ such that $\mathtt{multicomm}(G,L,T
)$ returns
$\mathtt{true}$. Both the multicommodity and maximum concurrent flow
problems are known to be computable in polynomial
time~\cite{mcfp,multi-commodity}.

\subsection{Algorithms}

\paragraph{$\greedya$}

A seemingly natural way to implement advanced channel reservation
is to follow a greedy procedure, where, at the time a request is
placed, the request is allocated a path (or set of paths)
guaranteeing the earliest possible completion time. We refer to
this approach as $\greedya$ and explain it next.

The $\greedya$ algorithm divides the time axis into slots delineated
by event. Each event corresponds to a set up or tear down instance
of a connection. Therefore, during each time slot the state of all
links in the network remains unchanged. In general, the time axis
will consist of $n$ time slots, where $n \ge 1$ is a variable and
slot $i$ corresponds to  time interval $[t_i,t_{i+1}]$. Note that
$t_1=t$ (the time at which the current request is placed) and
$t_{n+1}=\infty$.

Let $W_i = {\{ b_i(1), b_i(2),\ldots, b_i(|E|) }\}$  be the vector
of reserved bandwidth on all links at time slot $i$ where
$i=1,\ldots,n$, and $b_i(e)$
 denote the reserved bandwidth on link $e$ during slot $i$, with $e=1,\ldots,|E|$.

For each slot $i \in L$, construct a graph $G_i$, where the
capacity of link $e \in E$ is $C_i(e)=C(e)-b_i(e)$, i.e., $C_i(e)$
represents the available bandwidth on link $e$ during slot $i$.

In order to guarantee the earliest completion time, $\greedya$
repeatedly performs a max flow allocation between nodes $s$ and
$d$, for as many time slots as needed until the entire file is
transferred. This approach ensures that, in each time slot, the
maximum possible number of bits is transmitted and, hence, the
earliest possible completion time (at the time the request is
placed) is achieved. The $\greedya$ algorithm can thus be
concisely described with the following pseudo-code:

\begin{enumerate}
\item Initialization
\begin{itemize}
\item Set initial time slot: $i$=1.
\item Set initial size of remaining file: $r=f$.
\end{itemize}
\item If
$\mathtt{maxflow}(G_i,s,d,r) \geq t_{i+1}-t_i$ (i.e., all the
remaining file can be transferred during the current time slot) then
\item Exit step:
\begin{itemize}
\item Update $W_i$ by subtracting the used bandwidth from every link in the new
flow and, if needed, create a new event $t_{i+1}$ corresponding to
end of connection.
\item Exit procedure.
\end{itemize}
else \item Non-exit step:
\begin{itemize}
\item Update $W_i$ by subtracting the used bandwidth from every link in the new flow.
\item $r = r
-r\times(t_{i+1}-t_i)/\mathtt{maxflow}(G_i,s,d,r)$ (update size of
remaining file).
\item $i = i+1$ (advance to the next time slot).
 \item Go back to step 2.
\end{itemize}
\end{enumerate}

\paragraph{$\greedyb$}

The $\greedyb$ algorithm is a variation of $\greedya$, where only
shortest paths (in terms of number of hops) between the source $s$
and destination $d$ are utilized to route data. To implement
$\greedyb$, we employ exactly the same procedure as in $\greedya$,
except that we prune all the links not belonging to one of the
shortest paths using breadth first search. Note that for a given
source $s$ and destination $d$, the pruned links are the same for
all the graphs $G_i$.

\subsection{Inefficiency Results}

We next constructively show that there exist certain arrival
patterns, for which both $\greedya$ and $\greedyb$ achieve a
saturation throughput significantly lower than optimal.
Specifically, we present cases where the saturation throughput of
$\greedya$ and $\greedyb$ is $\Omega(|V|)$ times smaller than the
optimal throughput. Thus, the ratio of the maximum throughput of the
greedy algorithms to that of the optimal algorithm can be
arbitrarily small.

\begin{theorem}
For any given vertex set with cardinality $|V|$, there exists a graph
$G(V,E)$ such that the saturation throughput of $\greedya$ is
$|V|/2$ times smaller than the optimal saturation throughput.
\end{theorem}
\begin{proof}
Consider the ring network shown in Fig. \ref{fig:greed}(a).
Suppose every link is an undirected 1~Gb/s link and that during
each second requests arrive in this order: 1Gb request from node 1
to 2, 1Gb request from node 2 to 3, 1Gb request from node 3 to 4,
$\cdots$, 1Gb request from node $|V|$ to 1.

 The $\greedya$ algorithm will
allocate the maximum flow to each request, meaning it will be
split to two paths, half flowing directly, and half flowing
through the alternate path along the entire ring. Therefore, each
new request will have to wait to the previous one to end,
resulting in a completion time of $|V|/2$ seconds.

On the other hand,  the optimal time is just one second using the
direct link between each pair of nodes. Assuming the above pattern
of requests repeats periodically, then an optimal algorithm can
support at most one request per second between each pair of
neighboring nodes before reaching saturation, while $\greedya$ can
support at most $2/|V|$ request per second, hence proving the
theorem. A similar result can be proved for directed graphs.
\end{proof}

We next show that restricting routing to shortest paths does not
solve the inefficiency problem.

\begin{theorem}
For any given vertex set cardinality $|V|$, there exists a graph
$G(V,E)$ such that the saturation throughput of $\greedyb$ (or any
other algorithm using only shortest path routing) is $|V|-2$ times
smaller than the optimal saturation throughput.
\end{theorem}
\begin{proof}
Consider the network depicted in Fig. \ref{fig:greed}(b), where
all requests are from node 1 to 2, and only the direct path is
used by the algorithm. In this scenario, an optimal algorithm would
all $|V|-2$ paths between nodes 1 and 2. Hence, the optimal
algorithm can achieve a saturation throughput $|V|-2$ higher than
$\greedyb$.
\end{proof}

\begin{figure}
\epsfxsize=0.15\textwidth \epsfbox{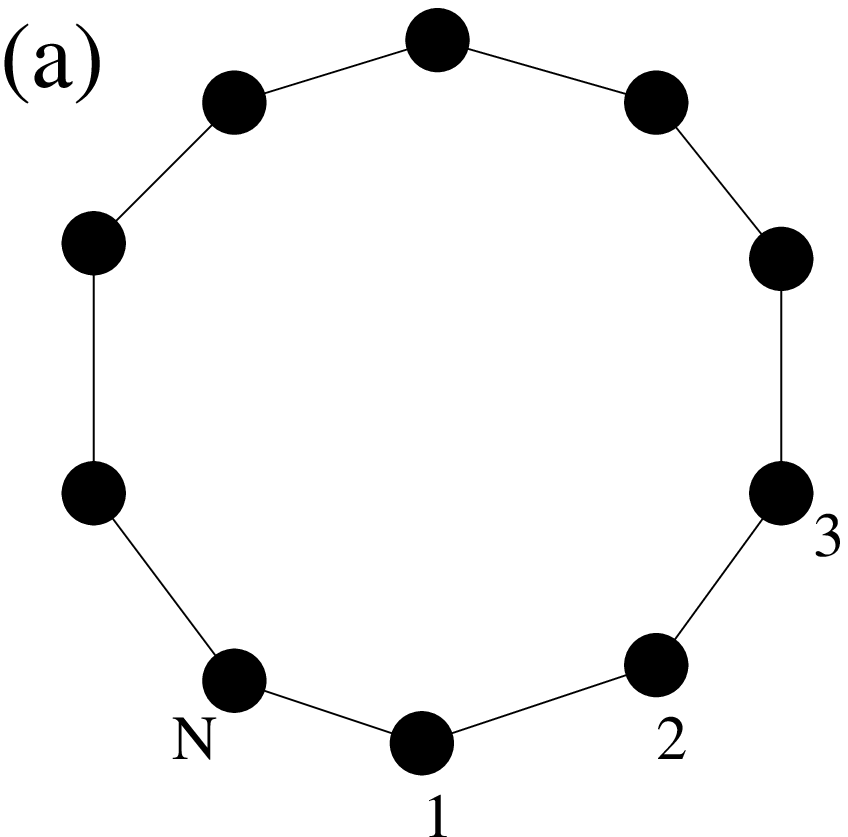}
\epsfxsize=0.3\textwidth \epsfbox{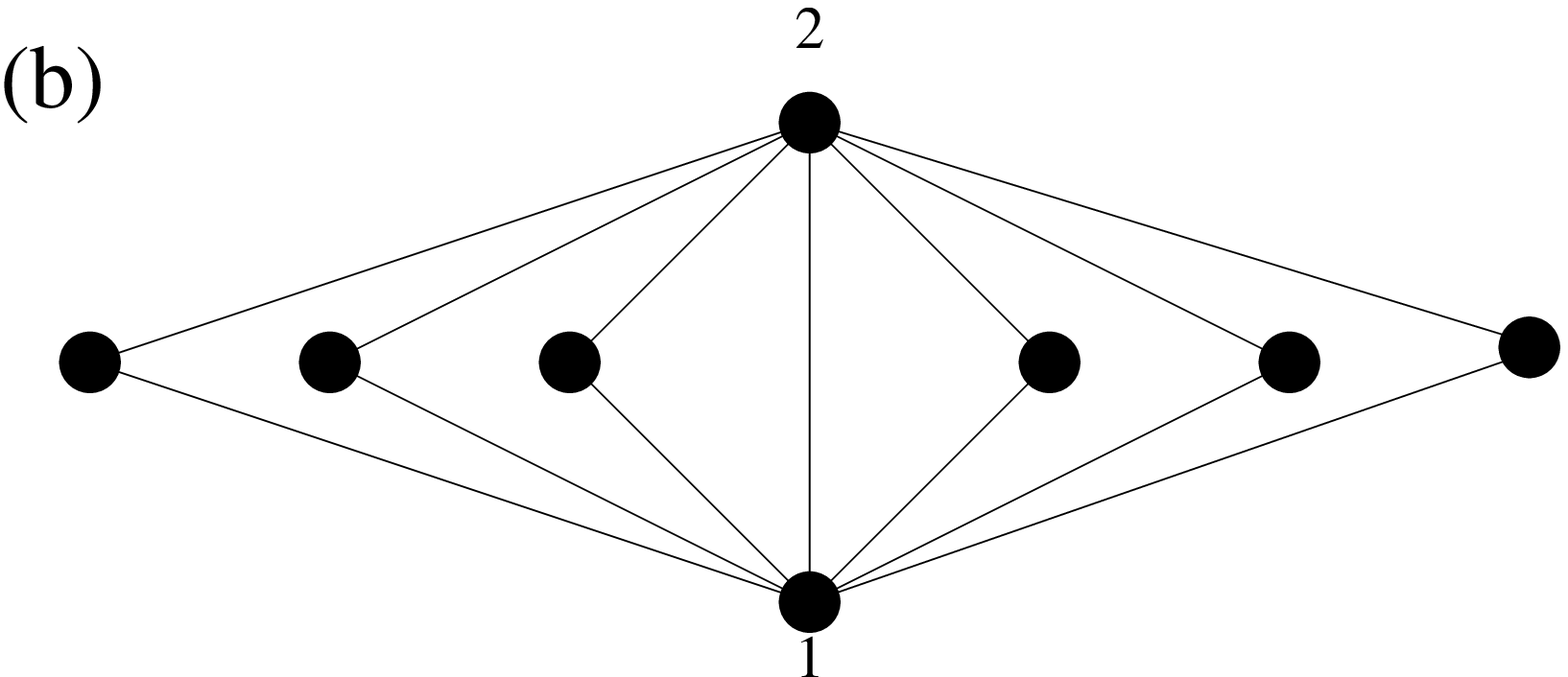} \caption{Examples for
the disadvantages of the greedy approach. \label{fig:greed}}
\end{figure}

\section{Competitive Algorithms}
\label{sec:comp}

In the previous section, we uncovered some of the limitations of
greedy algorithms that immediately allocate resources to each
arriving request. In this section, we present a new family of
on-line, polynomial-time algorithms that rely on the idea of
``batching'' requests. Thus, instead of immediately reserving a path
for each incoming request as in the greedy algorithms, we accumulate
several arrivals in a batch and assign a more efficient set of flow
paths to the whole batch. The proposed algorithms guarantee that the
maximum delay experienced by each request is within a finite
multiplicative factor of the value achieved with an optimal off-line
algorithm. Thus, these algorithms reach the maximum
throughput achievable by any algorithm.

%%%%%%%%%%%%%%%%%%%%%%%%%%
\subsection{Capacity Bounds}
\label{sec:fluid}
%%%%%%%%%%%%%%%%%%%%%%%%%%
As a preparatory step, we present a bound on the maximum
throughput achievable by any algorithm. This bound, based on a
multicommodity flow formulation, will be used to compare the
performance of the proposed on-line competitive algorithms to the
optimal off-line algorithm.
\begin{lemma}
\label{lemma:multicomm}
 If during any time interval $T$, each node $i
\in V$ sends on average $\gamma_{ij}$ bits of information per unit
time to node $j \in V-\{i \}$, then there exists a multi-commodity
flow allocation on graph $G$ with flow values $f_{ij} =
\gamma_{ij}$. \label{first}
\end{lemma}
\begin{proof}
First, we prove that for any given arrangement of sources and
destinations in a network and at any time $t$, the transmission
rates in bits per unit time comply with the properties of
multi-commodity flows. Take the transmission rate between node $i$
and $j$ going through the (directed) edge $k$ at time $t$ to be
$R^k_{i,j}(t)$. The total throughput passing each link can not
exceed the link capacity, $\sum_{i,j}R^k_{i,j}(t)\leq C(k)$.
According to the information conservation property, the total
transmission rate into any node is equal to the out-going rate
except for the source and destination of any of the transmissions.
Now, the average transmissions during any time interval $T$
satisfy the same properties as above because of linearity. Thus,
we obtain a valid multi-commodity flow problem, where the time
average of each flow represents an admissible commodity.
\end{proof}

\begin{corollary}
\label{multi_cor} The average transmission rate between all pairs
in the network for any time interval is a feasible multicommodity
flow.
\end{corollary}

\subsection{The $\competa$ Competitive Algorithm}

In case no deterministic knowledge on the future requests is
given, one would like to give some bounds on the performance of
the algorithm compared to the performance of a ``perfect''
off-line algorithm (i.e., with full knowledge of the future and
unlimited computational resources). We present here an algorithm,
called $\competa$ (since it batches together all pending
requests), giving bounds on this measure.

The algorithm can be described as follows (we assume that there is
initially no pending request):
\begin{enumerate}
\item For a request $l=\{s,d,f\}$
arriving at time $t$, give an immediate starting time and an ending
time of $t_{c}=t+\mathtt{maxflow}(G,l)$. \label{this0} \item Set $L
\leftarrow \mbox{null}$\label{this1}.
\item \label{this2} While $t<t_c$ and another request $l'=\{s',d',f'\}$ arrives at time $t$:
\begin{itemize}
\item Set $L \leftarrow L \cup l'.$
\item Mark $t_{c}$ as its starting time and add it
to the waiting batch.
\end{itemize}
\item When $t=t_c$, calculate $t'=\mathtt{maxflow}(G,L)$.
\begin{itemize}
\item If $t'=0$ (i.e., there is no pending request)  go back to step~\ref{this0}.
\item Else $t_c = t_c + t'$ and go back to step~\ref{this1}.
\end{itemize}
\end{enumerate}

We note that upon the arrival of a request, the $\competa$ algorithm
returns only the starting time of the connection. The allocated
paths and completion time are computed only when the connection
starts (in the next section, we present another competitive
algorithm with a slightly looser competitive ratio but which returns
the completion time at the arrival time of a request).

We next compare the performance of an optimal off-line algorithm
in the network with that of the $\competa$ algorithm in an
augmented network. The augmented network is similar to the
original network, other than that it has a capacity of
$(1+\epsi)C(e)$ at any link, $e$, that has capacity $C(e)$ in the
original network. This implies that the performance of the
competitive algorithm is comparable if one allows a factor $\epsi$
extra capacity in the links, or, alternatively, one may say that
the performance of the algorithm is comparable to the optimal
off-line algorithm in some lower capacity network, allowing the
maximum rate of only $(1+\epsi)^{-1}C(e)$ for each link. The
$\competa$ algorithm satisfies the following theoretical property
on the maximum waiting time:
\begin{theorem}
\label{thm:compet}
In a network with augmented resources where every edge has
bandwidth $(1+\epsi)$ times the original bandwidth
(for all $\epsi>0$), the maximum waiting time from request
arrival to the end of the transmission using $\competa$,
for requests arriving up to any time $t^*$, is no
more than $2/\epsi$ times the maximum waiting time for the
optimal algorithm in the original network.
\end{theorem}
\begin{proof}
Consider the augmented resource network. Take the maximum length
batch that accommodates requests arriving before time $t^*$, say
$i$, and mark its length by $T$. Since the batch before this one was
of length at most $T$ (if a request arrives when no batch is
executing, the batch is considered to be of size 0), and all
requests in batch $i$ were received during the execution of previous
batch $i-1$, then the total waiting time of each of these requests
was at most $2T$. By Corollary~\ref{multi_cor}, the total time for
handling all requests received during the execution of batch $i-1$
must have been at least $T$ in the augmented network, or
$(1+\epsi)T$ in the original network. Since all of these requests
arrived during a time of at most $T$, one of them must have been
continuing at least a time of $\epsi T$ after the last request under
any algorithm. Therefore the ratio between the maximum waiting time
is at most $2/\epsi$.
\end{proof}

\begin{corollary} The saturation throughput of $\competa$ is optimal because for
any arbitrarily small $\epsilon
> 0$, the delay of a request is guaranteed to be at most a
finite multiplicative factor larger than the maximum delay of the
optimal algorithm in the reduced resources network.
\end{corollary}

We next show that the computational complexity of $\competa$ is
practical.

\begin{theorem}
The computational complexity of algorithm $\competa$ per
request is at most polynomial in the size of the network.
\end{theorem}
\begin{proof}
For each batch, the $\competa$ algorithm solves a maximum concurrent
flow problem that can be computed in polynomial time~\cite{mcfp}.
 Noting that the number of
batches cannot exceed the number of requests (since each batch
contains at least one request), the theorem is proven.
\end{proof}

\subsection{The $\competb$ Competitive Algorithm}
The following algorithm, called $\competb$, operates in a similar
setup to $\competa$.  However, it gives a guarantee on the finishing
time of a job when it is submitted.

The algorithm is based on creating increasing size batches in case
of high loads. When the rate of request arrival is
high, each batch created will be approximately twice as long as the
previous one, and thus larger batches, giving a good approximation
of the total flow will be created. When the load is low only a single
batch of pending requests exists at most times, and its size will
remain approximately constant or even decrease for decreasing load.

The algorithm maintains a list of times, $t_i$, $i=1,\ldots,n$,
where for every interval, $[t_i,t_{i+1}]$, a batch of jobs is
assigned. When a new request between a source $s$ and a destination
$d$ arrives at time $t$, then an attempt is first made to add it to
one of the existing intervals by using a multicommodity flow
calculation. If the attempt fails, a new time, then a new batch is
created and  $t_{n+1}=\max(t_n+(t_n-t),t_n+M)$, where
$M=f/\mathtt{maxflow}(G,s,d)$ is the minimum time for the job
completion, is added to the list, and the job is assigned to the
interval $[t_n,t_{n+1}]$.

We next provide a detailed description of how the algorithm handles
a new request. We use the tuple $l=\{s,f,d\}$ to denote this job and
the list $L_i$ to denote the set of jobs already assigned to slot
$i$.

\begin{enumerate}
\item Initialization: set $i \leftarrow 1$.
\item While $i \leq n$
\begin{itemize}
\item If $\mathtt{multicomm}(G,L_i\cup l,t_{i+1}-t_{i}) =
\mathtt{true}$ then return $i$ and exit (check if request can be
accommodated during slot $i$).
\item $i \leftarrow i+1$.
\end{itemize}
\item Set $n \leftarrow n+1$ (create a new slot).
\item Set $M=\mathtt{maxflow}(G,l)$.
\item If $M>t_{n-1}-t$\\
then $t_{n+1}\leftarrow t_{n}+M$;\\
else $t_{n+1}\leftarrow t_{n}+(t_{n}-t)$.
\item Return $n$.
\end{enumerate}
Fig.~\ref{fig:compet2} illustrates runs of the $\competa$ and
$\competb$ algorithms for the same set of requests.

The advantage of the $\competb$ algorithm over $\competa$ is the
assignment of a completion time to a request upon its arrival.
Furthermore, the proof of the algorithm requires only an attempt to
assign a new job to the last interval, rather than to all of them.
Therefore, it is also possible to give the details of the selected
paths to the source as soon as a new batch is created.

We introduce the following technical Lemma which will be used to
prove the main theorem.
\begin{lemma}
\label{lem:interval}
(a) For any $i$, $t_i-t_{i-1}\geq(t_i-t)/2$.
(b) For any $i$, $t_{i+1}-t_i\leq \max(2(t_i-t_{i-1}),M)$, where $M$ is the
minimum completion time for the largest job in the interval $[t_{i+1},t_i]$.
\end{lemma}
\begin{proof}
(a) When the time slot $[t_{i-1},t_i]$ is added, its length is the maximum
of $M$ and $[t,t_{i-1}]$. As $t$ increases $t_i-t_{i-1}\geq t_{i-1}-t$, and
therefore $2(t_i-t_{i-1})\geq (t_i-t_{i-1})+ (t_{i-1}-t)=t_i-t$.

(b) From (a) follows that at the time $t$ when $[t_{i+1},t_i]$ is added
$t_i-t_{i-1}\geq(t_i-t)/2$. Since $t_{i+1}-t_i=\max(t_i-t,M)$ the Lemma
follows.
\end{proof}

\begin{theorem}
For every $\epsi>0$ and every time $t^*$,
the maximum waiting time from the request time
to the end of a job for any request arriving up to time
$t^*$ and
 for a network with augmented resources,
is no more than $4/\epsi$ times the maximum waiting time
for the optimal algorithm in the original network.
\end{theorem}
\begin{proof}
When creating a new batch, at the interval $[t_n,t_{n+1}]$ there are
two possibilities:
\begin{enumerate}
\item
$t_{n+1}=t_n+M$, where $M$ is the minimum time for the new job
completion. In this case, the minimum time for the completion of the
job in the original network would have been $(1+\epsi)M$. By
Lemma~\ref{lem:interval}(a) the total waiting time for the job here is
at most $2M$, and therefore the ratio of waiting times is
$2/(1+\epsi)<4/\epsi$.
\item
$t_{n+1}=t_n+(t_n-t)$,
in this case, there exists a request among the set of tasks in batch
$n$ that according to theorem~\ref{thm:compet} waits for a time of
at least $(t_n - t_{n-1})$. Therefore, the maximum waiting time in
the original network is at least $\epsi (t_n - t_{n-1})$.
 Since $t_n-t_{n-1}\geq(t_n-t)/2$, and
$t_{n+1}=t_n+t_n-t$, it follows that $t_{n+1}-t\leq
4(t_n-t_{n-1})$. Therefore, the ratio is at most $4/\epsi$.
\end{enumerate}
\end{proof}

\begin{figure}
\epsfxsize=0.4\textwidth \epsfbox{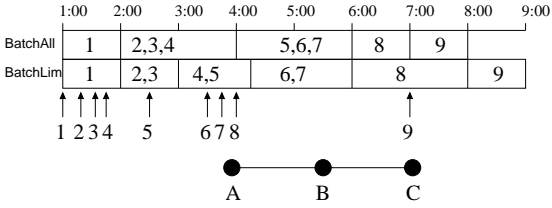} \caption{Illustration
of the batching process by the algorithms $\competa$ and $\competb$.
The odd numbered jobs request one hour of maximum bandwidth traffic
between nodes $A$ and $B$ and the even numbered jobs request one
hour maximum bandwidth traffic between $B$ and $C$. In $\competa$, a
new batch is created for all requests arriving during the running of
a previous batch. In $\competb$, for each new request an attempt is
made to add it to one of the existing windows, and if it fails, a
new window is appended at the end.
\label{fig:compet2} }
\end{figure}

\section{Bounding Path Dispersion}
\label{sec:dispersion}

\subsection{Approximating flows by single circuits}

The algorithms presented in the previous sections (both greedy and
competitive), do not limit the path dispersion, i.e., the number of
number of paths simultaneously used by a connection. In practice, it
is desirable to minimize the number of such paths due to the cost of
setting up many paths and the need to split a into many low capacity
channels. The following suggests a method of achieving this goal.

\begin{lemma}
\label{lem:path} In every flow of size $F$ between two nodes on a
directed graph $G(V,E)$ there exists a path of capacity at least
$F/|E|$ between these nodes.
\end{lemma}
\begin{proof}
Remove all edges of capacity (strictly) less than $F/|E|$
from the graph. The total capacity of these edges is smaller than
$F/|E|\times|E|=F$.  By the Max-Flow--Min-Cut theorem, the maximum
flow equals the minimum cut in the network. Therefore, since the
total flow is $F$ there must remain at least one path between the
nodes after the removal. All edges in this path have capacity of at
least $F/|E|$. Therefore, the path capacity is at least $F/|E|$.
\end{proof}

Lemma~\ref{lem:path} is tight in order in both $|E|$ and $|V|$ as
can be seen for the network in Fig.~\ref{fig:net_cut}. In this
network each path amounts for only $O(|E|^{-1})=O(|V|^{-2})$ of the
total flow between nodes $1$ and $2$.
Therefore, to obtain a flow that consists of a given
fraction of the maximum flow at least $O(|E|)=O(|V|^2)$ of paths
must be used.

\begin{figure}
\epsfxsize=0.3\textwidth
\hskip 0.1\textwidth
\epsfbox{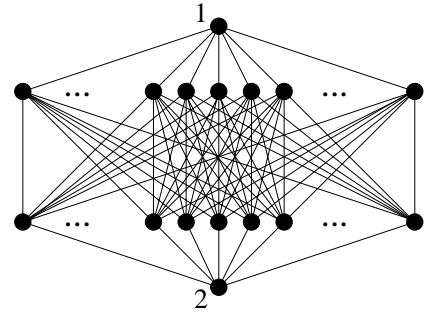}
\caption{A network demonstrating the optimality of Lemma~\ref{lem:path}.
\label{fig:net_cut}}
\end{figure}

The following theorem establishes the maximum number of paths
needed to achieve a throughput with a constant factor of that
achieved by the original flow.

\begin{theorem}
\label{thm:approx}
For every flow of size $F$ between two nodes
on a directed graph $G(V,E)$  and for every $\alpha>0$ there exists
a set of at most $\lceil\alpha|E|\rceil$ paths achieving a
flow of at least $(1-e^{-\alpha})F$.
\end{theorem}
\begin{proof}
Apply Lemma \ref{lem:path}, $\lceil\alpha|E|\rceil$ times. Each time
an appropriate path is found, its flow is removed from all of its links.
For each such path the remaining flow is multiplied by at most
$1-\frac{1}{|E|}\leq\exp\left(\frac{1}{|E|}\right)$.
\end{proof}

\subsection{Competitive Algorithms}
\label{sec:dispersion-alg}

Theorem \ref{thm:approx} provides both an algorithm for reducing the
number of paths and a bound on the throughput loss. To approximate the flow
by a limited number of paths remove all edges with capacity less than
$F/|E|$ and find a remaining path. This process can be repeated to
obtain improving approximations of the flow. The maximum number of
paths needed to achieve an approximation of the flow to within a
constant factor is linear in $|E|$, while the number of possible paths
may be exponential in $|E|$.

Using the above approximation in conjunction with the competitive
algorithms one can devise two new algorithms $\competc$ and
$\competd$, based upon $\competa$ and $\competb$ respectively.
These algorithms perform similarly to $\competa$ and $\competb$,
in terms of the batching process. However, after filling each
batch, the algorithms will limit the dispersion of each flow, and
approximate the flow. To achieve a partial flow, each time we
select the widest path from remaining edges (where the weights are
determined by the link utilization in the solution of the
multicommodity flow) and reserve the total path bandwidth. We
repeat this until either the desired number of paths is reached or
the entire flow is routed. Setting the dispersion bound to $\alpha
|E|$, Theorem \ref{thm:approx} guarantees that the time allocated
for the slot will be no more than
$\mathtt{maxflow}(G,L)/(1-e^{-\alpha})$, where
$\mathtt{maxflow}(G,L)$ is the original slot duration. Algorithms
$\competc$ and $\competd$ will achieve a throughput within a
factor of $1-e^{-\alpha}$ of the maximum throughput (and will no
longer be throughput optimal as $\competa$ and $\competb$, but
only throughput competitive).

\section{Simulations}
\label{sec:simulations}
In this section, we present simulation results illustrating the
performance of the algorithms described in this paper. While the
emphasis of the previous sections was on the throughput optimality
(or competitiveness) of our proposed algorithms, here we are also
interested in evaluating their performance for other metrics,
especially average delay. Average delay is defined as the average
time elapsing from the point a request is placed until the
corresponding connection is completed.

The main points of interest are as follows: (1) how do the
competitive algorithms fare compared to the greedy approach of
Section~\ref{sec:greedy}? (2) what value of path dispersion is
needed to ensure good performance? It is important to emphasize
that we do not expect the throughput optimal competitive
algorithms to always outperform the greedy approach in terms of
average delay.

%%%%%%%%%%%%%%%%%%%%%%%%%%%%%%%%%%%%%%%%%%%%%%%%%%%%%%%%%%
\subsection{Simulation Set-Up}
\label{sec:sim_par}
%%%%%%%%%%%%%%%%%%%%%%%%%%%%%%%%%%%%%%%%%%%%%%%%%%%%%%%%%%

We have written our own in simulator in C++. The simulator uses the
COIN-OR Linear Program Solver (CLP) library \cite{coin} to solve
multi-commodity optimization problems and allows evaluating our
algorithms under various topological settings and traffic
conditions. The main simulation parameters are as follows:

\begin{figure}[t]
  \begin{center}
%    \subfigure[4 node]{
%    \resizebox{0.55in}{!}{\includegraphics{topology1}}}
%    \hspace{0.4in}
    \subfigure[8 node]{
    \label{subfig:topology1}
    \resizebox{0.75in}{!}{\includegraphics{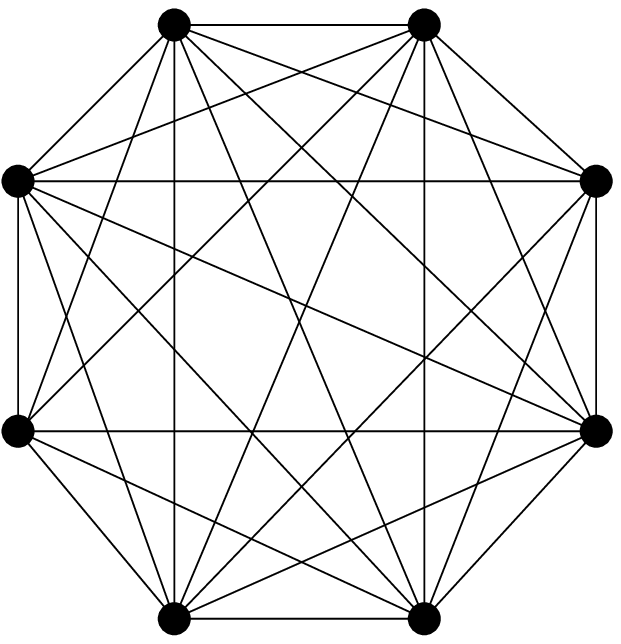}}}
    \hspace{0.2in}
    \subfigure[11 node LambdaRail]{
    \label{subfig:topology2}
    \resizebox{2in}{!}{\includegraphics{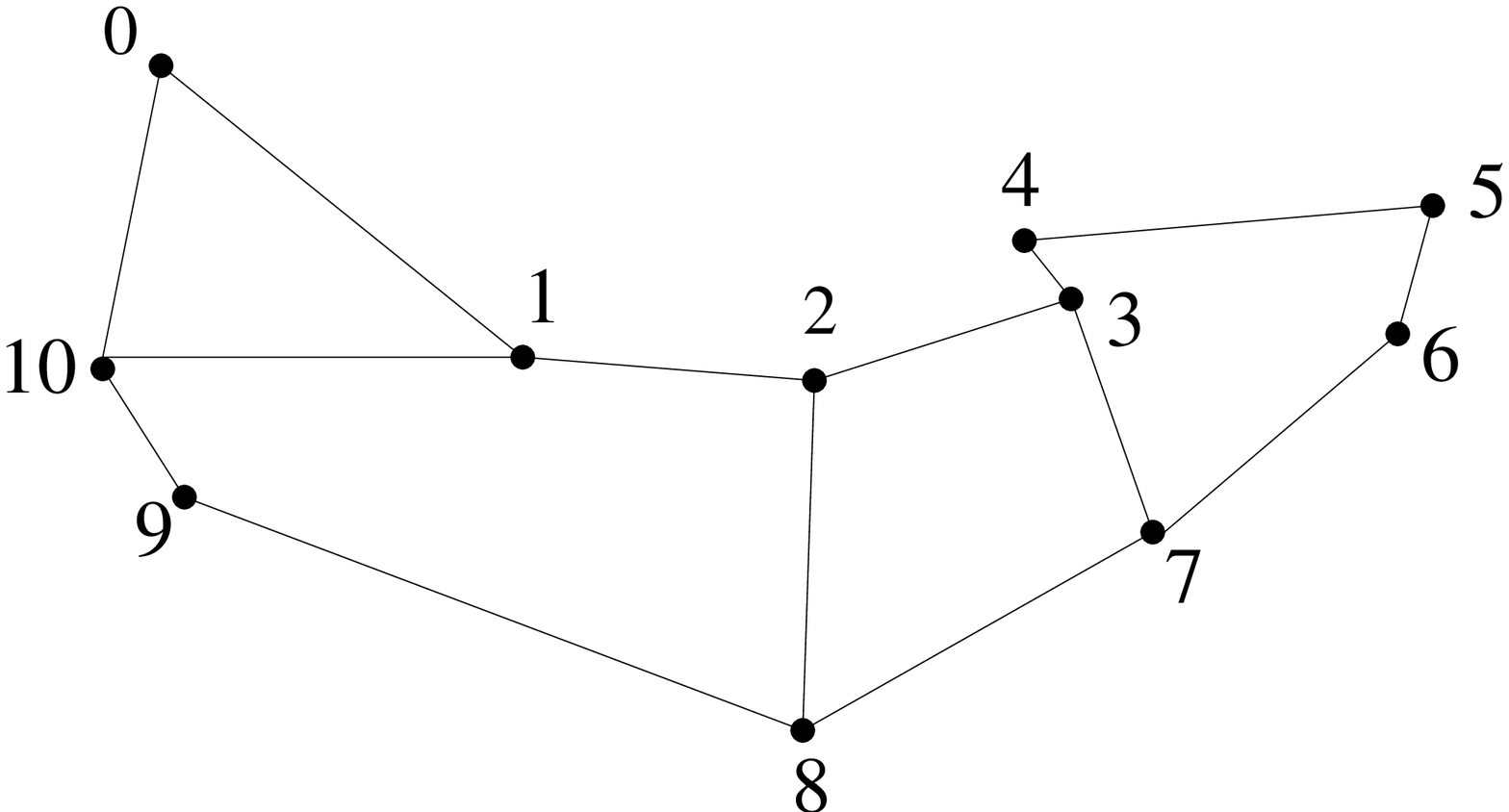}}}
    \caption{Simulation topologies.}
  \label{fig:topologies}
  \end{center}
\end{figure}
\begin{itemize}

\item \emph{Topology:} our simulator supports arbitrary
topologies. In this paper, we consider the two topologies depicted
in Figure~\ref{fig:topologies}. One is a fully connected graph
(clique) of eight nodes and the other is an 11-node topology,
similar to the National LambdaRail testbed~\cite{nlr}. Each link on
these graphs is full-duplex and assumed to have a capacity of 20
Gb/s.
\item \emph{Arrival process:} we assume that the aggregated
arrival of requests to the network forms a Poisson process (this can
easily be changed, if desired). The mean rate of arrivals is
adjustable. Our delay measurements are carried out at different
arrival rates, referred to as \emph{network load}, in units of
requests per hour.
\item \emph{File size distribution:}

We consider two models for the file size distribution:
\begin{enumerate}
\item Pareto:
$$F(x) = 1- \left( \frac {x_m}{x-\gamma} \right)^{\beta} \mbox{, where } x \ge x_m+\gamma.$$
In the simulations, we set $\beta = 2.5$, $x_m = 1.48$ TB
(terabyte) and $\gamma = 6.25*10^{-3}$ TB, implying that the mean
file size is $2.475$ TB. \item Exponential:
$$F(x) = \exp(-\lambda x) \mbox{, where } x \ge 0.$$
In the simulations, $1/\lambda=2.475$ TB.
\end{enumerate}

\item\emph{Source and Destination}: for each request, the source and
destination are selected uniformly at random, except that they must
be different nodes.

\end{itemize}
All the simulations are run for a total of at least $10^6$ requests.

%%%%%%%%%%%%%%%%%%%%%%%%%%%%%%%%%%%%%%%%%%%
\subsection{Results}
\label{sec:sim_res}
%%%%%%%%%%%%%%%%%%%%%%%%%%%%%%%%%%%%%%%%%%%

We first present simulation results for the clique topology. Figure
\ref{fig:par8-infocom} depicts the average delay as a function of
the network load for the algorithms $\competa$ and $\greedya$,
assuming a Pareto file size distribution. The figure indicates that
$\greedya$'s average delay increases sharply around
140~requests/hour, while $\competa$  can sustain higher network load
since it allocates flows more efficiently. On the other hand, at
lower load, $\greedya$ achieves a lower average delay than
$\competa$. This result can be explained by the fact that $\greedya$
does not wait to establish a connection. This results illustrates
the existence of a delay-throughput trade-off between the two
algorithms.

\begin{figure}[t]
\centering \resizebox{2.9in}{!}{\includegraphics{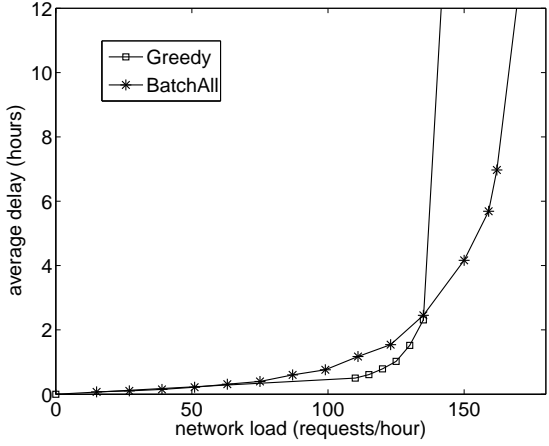}}
\caption{Performance comparison of algorithms $\competa$ and
$\greedya$ for the 8-node clique and Pareto file size distribution.}
\label{fig:par8-infocom}
\end{figure}

Figure \ref{fig:exp8-infocom-1}  compares the performance of the
$\competa$ and $\competb$ algorithms. The file size distribution
is exponential. The figure shows that $\competb$ is less efficient
than $\competa$ in terms of average delay,  especially at low
load. This result is somewhat expected given that $\competb$ uses
a less efficient batching process and its delay ratio is looser.
However, since $\competb$ is throughput optimal, its performance
approaches that of $\competa$ at high load.

\begin{figure}[t]
\centering \resizebox{2.9in}{!}{\includegraphics{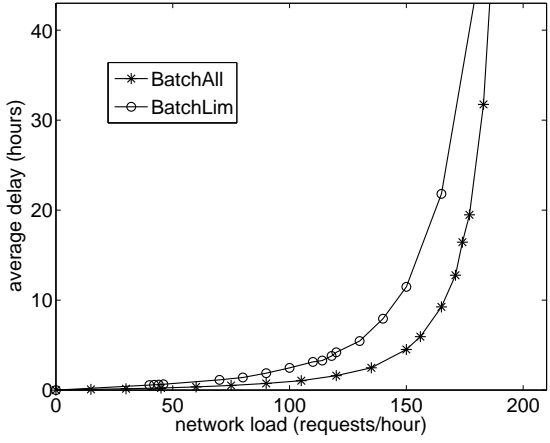}}
\caption{Performance comparison of algorithms $\competa$ and $\competb$
 for the 8-node clique and exponential file size distribution.}
 \label{fig:exp8-infocom-1}
\end{figure}

Figure \ref{fig:exp8-infocom-2}  evaluates and compares the
performance of the $\competa$ and $\competc$ algorithms. For
$\competc$, results are presented for the cases where the path
dispersion bound is set to either 1, 3, or 5 paths per connection.
The file size distribution is exponential. The figure shows also a
fluid bound on capacity derived in~\cite{FCS07}. Its value
represents an upper bound on the maximum network load for which
the average delay of requests is still bounded.

\begin{figure}[t]
\centering \resizebox{2.9in}{!}{\includegraphics{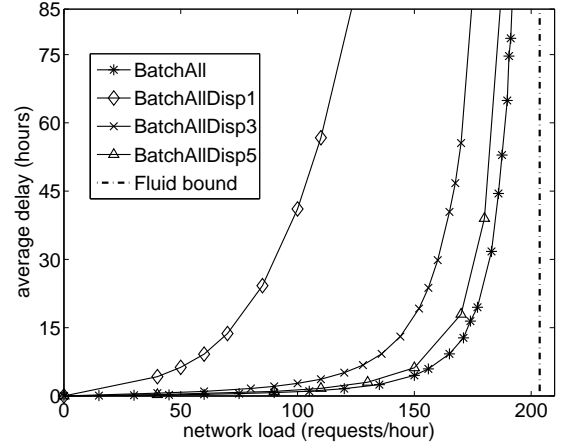}}
\caption{Performance evaluation of algorithms $\competa$
and $\competc$ with bounds 1, 3 or 5 paths per connection for the
8-node clique and exponential file size distribution. A fluid bound on
capacity is also depicted.}
 \label{fig:exp8-infocom-2}
\end{figure}

The figure shows that $\competa$ approaches the capacity bound at
a reasonable delay value and that a path dispersion of~$5$ per
 connection (corresponding to $\alpha = 0.089$) is sufficient for
$\competc$ to achieve performance close to $\competa$. It is worth
mentioning that, in this topology, there exist $1957$ possible
paths between any two nodes. Thus, with~5 paths, $\competc$ uses
only $0.25 \%$ of the total paths possible. The figure also
demonstrates the importance of multi-path routing: the performance
achieved using a single path per connections is far worse.

\begin{figure}[t]
\centering \resizebox{2.9in}{!}{\includegraphics{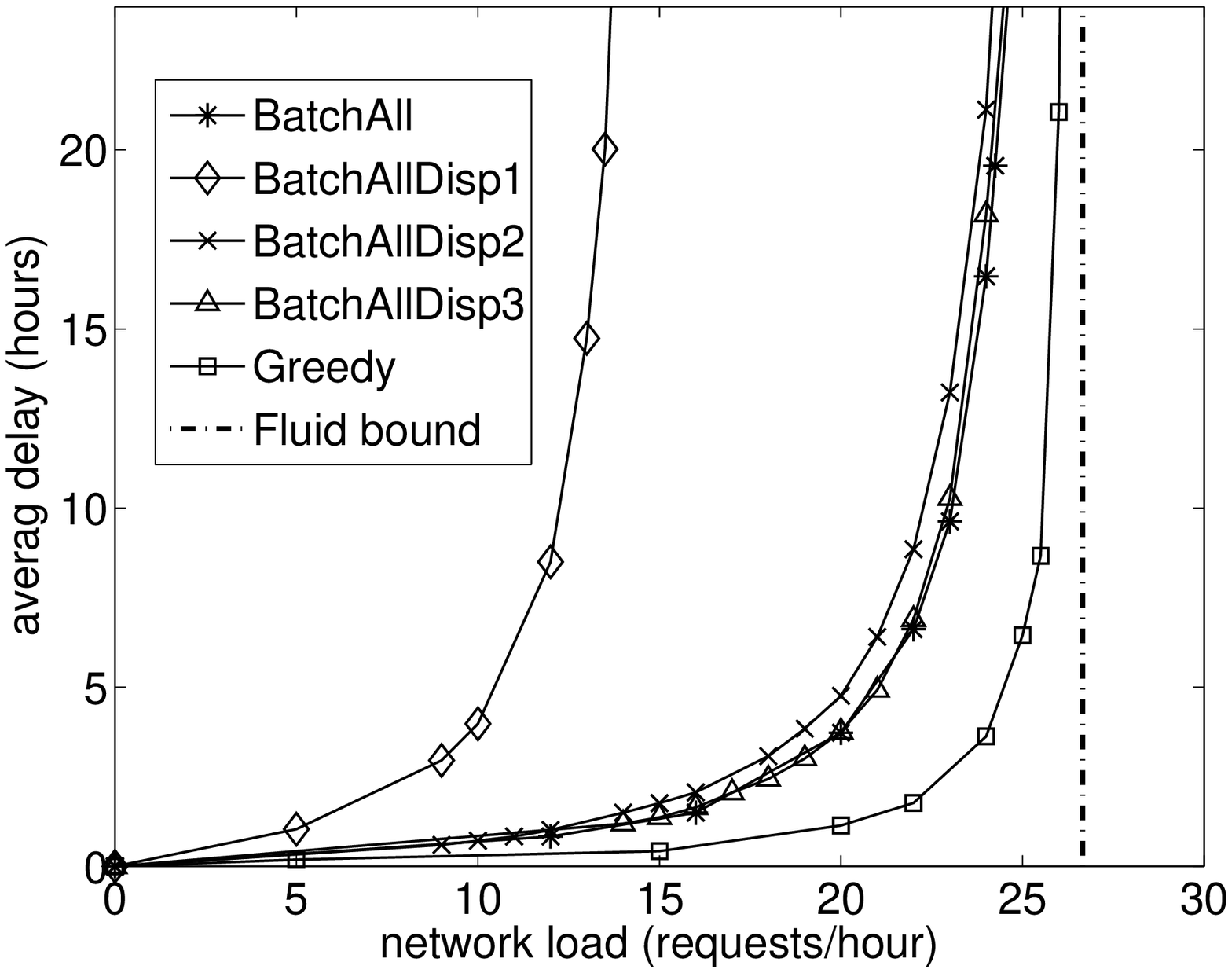}}
\caption{Performance evaluation of algorithms $\competa$, $\competc$
and $\greedyb$ for the 11-node topology \ref{subfig:topology2} and Pareto file size.
Algorithm $\competc$ is plotted with bounds of 3, 2
 and 1 on path dispersion.} \label{fig:par11-infocom}
\end{figure}

Figure \ref{fig:par11-infocom} depicts the performance of the
various algorithms and the fluid bound for the 11-node topology of
figure ~\ref{subfig:topology2}. The file size follows a Pareto
distribution. In this case,  we observe that $\competc$ with as
few as 3 paths per connection (or $\alpha = 0.107$) approximates
$\competa$ very closely. Since this network is sparser than the
previous one, it is reasonable to obtain good performance with a
smaller path dispersion. Another observation is that for this
topology and the range of network loads under consideration,
$\greedya$ outperforms $\competa$ in terms of average delay.
Although $\competa$ achieves the optimal saturation throughput,
its superiority over  $\greedya$ may only happen at very large
delays. The advantage of $\greedya$ in this scenario can be
explained by the sparsity of the underlying topology. Seemingly,
the added efficiency in throughput achieved with batching is not
sufficient to compensate the cost incurred by delaying requests.

\section{Conclusion}
\label{sec:summary} Advance reservation of dedicated network
resources is emerging as one of the most important features of the
next-generation of network architectures. In this paper, we have
made several advances in this area. Specifically, we have proposed
two new on-line algorithms for advanced reservation, called
$\competa$ and $\competb$, that are guaranteed to achieve optimal
throughput performance. We have proven that both algorithms
achieve a competitive ratio on the maximum delay until job
completion in $1+\epsi$ bandwidth augmented networks. While
$\competb$ has a slightly looser competitive ratio than that of
$\competa$ (i.e., $4/\epsilon$ instead of $2/\epsilon$), it has
the distinct advantage of computing the completion time of a
connection immediately as a request is placed.

A key observation of this paper is that path dispersion is
essential to achieve full network utilization. However, splitting
a transmission into too many different paths may render a
flow-based approach inapplicable in many real-world environments.
Thus, we have presented a rigorous, theoretical approach to
address the path dispersion problem and presented a method for
approximating the maximum multicommodity flow using a limited
number of paths. Specifically, while the number of paths between
two nodes in a network scales exponentially with the number of
edges, we have shown that throughput competitiveness up to any
desired ratio factor can be achieved with a number of paths
scaling linearly with the total number of edges. In practice, our
simulations indicate that 3 paths (in sparse graphs) to 5 paths
(in dense graphs) may be sufficient.

As shown in our simulations, for some topologies, a greedy
approach may perform better than the competitive algorithms,
especially at low load, in sparse topologies, or with uncorrelated
traffic between different pairs. However, we have shown that there
exist scenarios where greedy strategies can be highly inefficient.
This can never be the case for the competitive algorithms.

We conclude by noting that the algorithms proposed in this paper
can be either run in a centralized fashion (a reasonable solution
in small networks) or using link-state routing and distributed
signaling mechanism, such as enhanced versions of
GMPLS~\cite{dragon} or RSVP~\cite{ATM}. Distributed approximations
of the multicommodity flow have also been discussed extensively in
the literature (see, e.g., \cite{AL94,KPP95,AKS07}). Part of our
future work will be to investigate these implementation issues
into more detail.

\bibliographystyle{IEEEtran}
\bibliography{arxiv}

\end{document}